\documentclass[11pt]{article}
\usepackage{latexsym}

\usepackage[top=1in, bottom=1in, left=1in, right=1in]{geometry}

\usepackage{adjustbox}

\usepackage{amssymb}

\usepackage{cite}

\usepackage{graphicx} 
\usepackage{color} 
\usepackage{amsmath, amsthm, amssymb, appendix,ulem}
\usepackage{enumerate}
\usepackage{float}
\usepackage{url}
\usepackage{stackrel}
\usepackage{mathrsfs,dsfont}
\usepackage{caption}
\usepackage{subcaption}
\usepackage{wrapfig}
\usepackage{bbm}
\usepackage{bm}
\usepackage{epigraph}
\usepackage{physics}
\usepackage[colorinlistoftodos, shadow]{todonotes}
\usepackage{multirow}
\usepackage{hyperref}

\usepackage{mdwlist}

\usepackage{pifont}

\newtheorem{theorem}{Theorem}[section]

\newtheorem{lemma}[theorem]{Lemma}

\newtheorem{remark}[theorem]{Remark}

\theoremstyle{definition}

\theoremstyle{definition}

\theoremstyle{definition}

\newcommand{\I}{\mathcal{I}}
\newcommand{\F}{\mathcal{F}}
\newcommand{\G}{\mathcal{G}}
\newcommand{\W}{\mathcal{W}}

\newcommand{\E}{\mathbb{E}}
\newcommand{\PP}{\mathbb{P}}

\usepackage{tikz}    
\usetikzlibrary{shapes,automata,positioning,arrows.meta,fit}
\usepackage{mathtools}  
\usetikzlibrary{snakes}

\usepackage{caption}
\usepackage{tikz-cd}
\usetikzlibrary{babel}

\usepackage{tkz-euclide}

\numberwithin{equation}{section}

\usepackage{tikz-3dplot}

\usepackage{relsize}

 \tikzset{every node/.style={auto}}
 \tikzset{every state/.style={rectangle, minimum size=0pt, draw=none, font=\Large}}
  \tikzset{bend angle=7}
  
  \usepackage{array}
  
  \usepackage{chemfig}

\usepackage{authblk}

\makeatletter
\newcommand{\xrightleftarrows}[2]{%
  \mathrel{\mathop{%
    \vcenter{\offinterlineskip\m@th
      \ialign{\hfil##\hfil\cr
        \hphantom{$\scriptstyle\mspace{8mu}{#1}\mspace{8mu}$}\cr
        \bigarrowfill\cr
        \vrule height0pt width 2em\cr
        \bigarrowfill\cr
        \hphantom{$\scriptstyle\mspace{8mu}{#2}\mspace{8mu}$}\cr
        \noalign{\kern-0.3ex}
      }%
    }%
  }\limits^{#1}_{#2}}%
}
\makeatother

\begin{document}

\title{A disease-spread model on hypergraphs with distinct droplet and aerosol transmission modes}

\author[1]{Tung D. Nguyen} 
\affil[1]{Department of Mathematics, University of California Los Angeles, CA, USA}
\author[2]{Mason A. Porter}
\affil[2]{Department of Mathematics, University of California Los Angeles, CA, USA; Department of Sociology, University of California Los Angeles, CA, USA; Santa Fe Institute, Santa Fe, NM, USA}

\date{}
\maketitle


\begin{abstract}

We examine the spread of an infectious disease, such as one that is caused by a respiratory virus, with two distinct modes of transmission. To do this, we consider a susceptible--infected--susceptible (SIS) disease on a hypergraph, which allows us to incorporate the effects of both dyadic (i.e., pairwise) and polyadic (i.e., group) interactions on disease propagation. This disease can spread either via large droplets through direct social contacts, which we associate with edges (i.e., hyperedges of size 2), or via infected aerosols in the environment through
hyperedges of size at least 3 (i.e., polyadic interactions). 
We derive mean-field approximations of our model for two types of hypergraphs, and we obtain threshold conditions that characterize whether the disease dies out or becomes endemic. Additionally, we numerically simulate our model and a mean-field approximation of it to examine the impact of various factors, such as hyperedge size (when the size is uniform), hyperedge-size distribution (when the sizes are nonuniform), and hyperedge-recovery rates (when the sizes are nonuniform) on the disease dynamics.

\vspace{.1in}
\noindent \textbf{Relevance to Life Sciences}
\vspace{.1in}

In our paper, we formulate a disease-spread model on networks that distinguishes explicitly between two distinct modes of transmission: (1) spread via large droplets (which involve direct social contacts); and (2) spread via aerosols in the environment. Our model includes several important biological features and allows us to draw a variety of relevant biological conclusions.
First, we separate the effects of the two modes of transmission on the disease dynamics. Second, we obtain useful insights that depend on the structure of the hypergraph on which a disease spreads.
For hypergraphs with uniform hyperedge sizes (i.e., group sizes), we find that reducing this size may be the most effective way to eradicate the examined disease.
For hypergraphs with nonuniform hyperedge sizes, we find that having a few very large groups and many small groups (i.e., a heavy-tailed hyperedge-size distribution) leads to longer times before disease eradication (or even cause the disease to become endemic) in comparison to having uniform, medium-sized groups.
Finally, if one wants to combat the disease by increasing group recovery rates (e.g., by improving ventilation or air filtration), we find that increasing these rates in a way that depends on group sizes is just as effective
but possibly less costly than increasing them uniformly across all groups.

\vspace{.1in}
\noindent\textbf{Mathematical Content}
\vspace{.1in}

We formulate a stochastic susceptible--infected--susceptible (SIS) model on a hypergraph in which both nodes and groups (i.e., hyperedges of size at least 3) have associated state variables. To date, researchers have analyzed very few network models in which both nodes and hyperedges --- as opposed to only nodes --- have their own state variables.
We derive mean-field approximations of our disease-spread model for complete uniform hypergraphs and regular uniform hypergraphs, and we establish threshold conditions for the SIS disease to persist 
for these approximations. Guided by our mean-field approximations, we perform individual-level stochastic simulations to study how various network parameters and other model parameters affect the disease dynamics. We illustrate that both dyadic interactions and polyadic interactions play an important role in the disease dynamics.

\end{abstract}



\section{Introduction}

Mathematical modeling is an important approach to help understand, forecast, and control the spread of infectious diseases \cite{brauer2017mathematical,brauer2019}. Mathematical modeling of disease dynamics has a long history that goes back to the development of deterministic compartmental models about 100 years ago~\cite{kermack1927contribution}. Since then, there have been numerous developments in the modeling of disease dynamics, including the incorporation of stochastic effects~\cite{allen2008introduction,andersson2012stochastic}, the use of network structure to account for contacts and interactions between individuals~\cite{keeling2005networks,pastor2015epidemic,kiss2017,porter2016}, and the consideration of application-oriented details (such as {physical distancing~\cite{sun2011effect}, face-mask use~\cite{eikenberry2020mask}, and building ventilation~\cite{gao2016building}}).

Traditionally, one can model disease spread on a graph with nodes that represent individuals and edges that encode dyadic (i.e., pairwise) interactions between individuals \cite{ganesh2005effect,wang2003epidemic}. Recently, there has been growing interest in using hypergraphs to incorporate polyadic interactions (which are sometimes also called ``higher-order" interactions and often encode group interactions) in network models of disease spread~\cite{battiston2020networks,bick2023higher,bodo2016sis, cisneros2021multigroup, higham2021epidemics, landry2020effect}. See \cite{battiston2020networks,battiston2021physics,majhi2022dynamics,bick2023higher,battiston2025} for wide-ranging discussions of dynamical processes on polyadic networks.

One can examine disease spread on hypergraphs in several ways~\cite{battiston2025}. Some researchers have examined such dynamics
without distinguishing between transmission via dyadic and polyadic interactions~\cite{bodo2016sis,higham2021epidemics}.
Other researchers have distinguished between dyadic and polyadic interactions by using different infection rates for the two types of interactions~\cite{landry2020effect}. However, for some diseases, there are more fundamental differences between their propagation through dyadic and polyadic channels. 
For example, respiratory viruses like COVID-19 can spread through large droplets during close contact between individuals or through airborne transmission with infected aerosols \cite{domingo2020influence}. These two modes of transmission have distinct mechanisms, and they are affected by different human and environmental factors \cite{wang2021airborne}. Some disease-mitigation policies, such as different types of non-pharmaceutical interventions, can be more effective against one infection mode than the other mode.
For instance, physical distancing alone can help reduce transmission through droplets, but it may be less effective at preventing disease transmission via aerosols in poorly ventilated buildings \cite{pei2021human}.  

In the present paper, we formulate and analyze a disease-spread model on hypergraphs with two distinct transmission modes.
In our model, size-2 hyperedges (i.e., ordinary edges) encode dyadic interactions and represent direct contacts between individuals, and size-$s$ hyperedges with $s \geq 3$
encode polyadic interactions and represent a physical environment like a household or a workplace (rather than a social group of individuals).
We associate the droplet mode of disease transmission with size-2 hyperedges and the aerosol mode of disease transmission with all other hyperedges. Furthermore, in addition to assigning states (susceptible or infected) to the nodes, we also assign states (contaminated or uncontaminated) to the hyperedges to describe situations in which an environment has a high or low concentration of infected aerosols.
The idea of assigning states to environments has been considered in compartmental models in epidemiology (e.g., see \cite{feng2013mathematical,garira2014mathematical,lanzas2020modelling}). In the present paper, we
extend this idea to network models. 
Our assignment of states to hyperedges is also inspired by recent work on opinion dynamics on hypergraphs \cite{sampson2024oscillatory} and on synchronization on simplicial complexes \cite{millan2025topology}. 

We study our stochastic disease-spread model both analytically and numerical to gain insights into its dynamics. We derive mean-field approximations of it for two different types of hypergraphs, and we derive
expressions for basic reproduction number $R_0$ for both situations. We thereby examine when a disease persists (which occurs when
$R_0 > 1$) and when it dies out (which occurs when $R_0 < 1$). {In both mean-field models, the expression for $R_0$} is the sum of two terms that arise from separate contributions of the two disease-transmission modes.
In direct numerical simulations of our stochastic disease-spread model, we also observe some separation of the two transmission modes, although it is weaker than in the mean-field approximations.

Using both direct numerical simulations and our mean-field approximations, we
study the effects of hyperedge sizes and hyperedge-recovery rates (i.e., the rates at which contaminated environments become uncontaminated with infected aerosols). For networks with uniform hyperedge sizes, reducing this size is a very effective way to reduce the basic reproduction number $R_0$. 
For networks with nonuniform hyperedge sizes, we observe {in our numerical simulations} that
having a few very large hyperedges can cause disease eradication to take longer (and can even lead to disease persistence) than
having uniform, medium-sized hyperedges. 
Finally, we find that nonuniformly increasing hyperedge-recovery rates (e.g., by improving ventilation or air filtration) as a function of hyperedge sizes can be an effective way to mitigate disease spread.

Our paper proceeds as follows. In Section \ref{sec:model}, we introduce our stochastic disease-spread model and discuss several random-hypergraph models on which we simulate the disease dynamics. In Section \ref{sec:meanfield}, we derive two mean-field approximations of our model and establish threshold conditions for the examined disease to become extinct or to persist. 
In Section \ref{sec:simulation}, we perform numerical simulations on our stochastic model to test the accuracy
of our mean-field approximations and to investigate the effect of network structure on disease dynamics. 
Finally, in Section \ref{sec:discussion}, we summarize our findings and discuss both potential implication of our results and future directions. Our code is available at \url{https://github.com/TungDaoNguyen/SIS_Hypergraph}.


\section{A stochastic disease model on hypergraphs with two modes of transmission}\label{sec:model}

In this section, we describe our susceptible--infected--susceptible (SIS) disease model on hypergraphs with two modes of transmission. We review some basic notions about hypergraphs in Section~\ref{hyper}, present our model of disease dynamics on hypergraphs in Section~\ref{sec:infection}, and discuss the random-hypergraph models that we employ in Section~\ref{sec:hypergraph}.

\subsection{Hypergraphs} \label{hyper}

Consider an unweighted hypergraph $(V,E)$. The set $V$ of nodes represent individuals, and suppose that the hyperedge set $E = E_d\cup E_e$ 
has hyperedges of two disjoint categories, which correspond to two distinct transmission modes in our model of disease dynamics.
The set $E_d$ consists of hyperedges of size 2. These hyperedges, which encode pairwise (i.e., dyadic) interactions, represent direct contacts between two individuals. We associate these hyperedges with disease transmission via large droplets. The set $E_e$ consists of hyperedges of size at least 3. These hyperedges, which encode polyadic interactions, represent indoor physical environments (such as households, workplaces, and school classrooms) that individuals occupy. We associate these hyperedges with disease transmission via infected airborne aerosols.

Let $N = |V|$ denote the number of nodes of a hypergraph, and let $M = |E|$ denote the hypergraph's number of hyperedges. Henceforth, we refer to the hyperedges in the subset $E_d$ as ``edges" and to hyperedges in the subset $E_e$ as
``hyperedges". With this terminology in mind, we let $M_d = |E_d|$ denote the number of edges and let $M_e = |E_e|$ denote the number of hyperedges. 

For convenience, we describe a hypergraph using an incidence matrix $\I$, which is an $N\times M$ matrix such that $\I_{ih} = 1$ if node $i$ is in hyperedge $h$ and $\I_{ih} = 0$ otherwise \cite{newman2018networks}. Additionally, we use $\I^d$ to denote the incident matrix with respect to the edges and $\I^e$ to denote the incident matrix with respect to the hyperedges. The $N\times N$ matrix $\W^d = \I^d(\I^d)^T$ has entries $\W^d_{ij}$ that record the number of edges that include nodes $i$ and $j$, and the $N\times N$ matrix $\W^e = \I^e(\I^e)^T$ records the numbers of hyperedges that include each pair of nodes.


\subsection{Model of disease spread}\label{sec:infection}

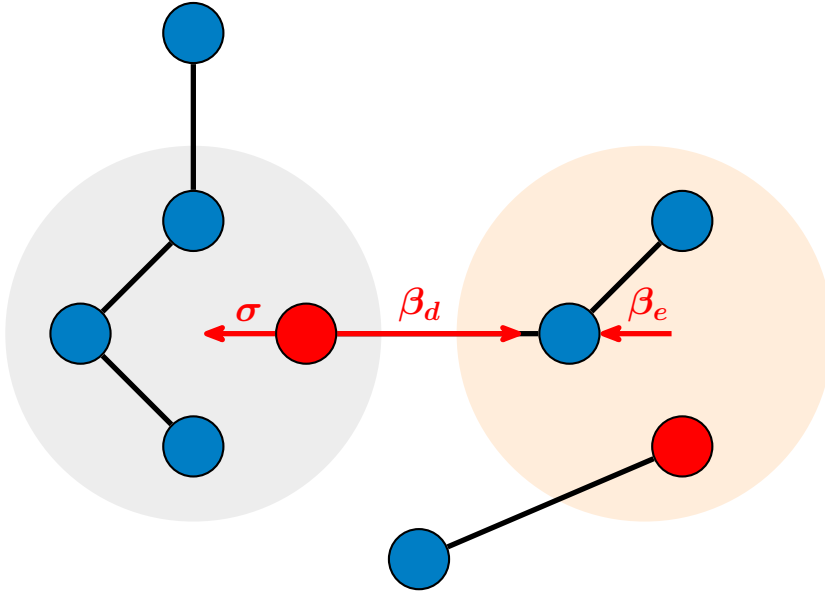
\begin{figure}
    \centering

\begin{tikzpicture}[thick, >={Stealth[round]}, node distance=1cm]


\fill[gray!20, opacity=0.7, on background layer] (-3,0) circle (2.5);      
\fill[orange!20, opacity=0.7, on background layer] (3,0) circle (2.5);     

\node (a0) at (-3,0) {}; 
\node[circle, draw, fill=cyan!60!blue, minimum size=8mm] (a1) at (-4.5,0) {};
\node[circle, draw, fill=cyan!60!blue, minimum size=8mm] (a2) at (-3,1.5) {};
\node[circle, draw, fill=cyan!60!blue, minimum size=8mm] (a3) at (-3,-1.5) {};
\node[circle, draw, fill=red, minimum size=8mm] (a4) at (-1.5,0) {}; 
\node[circle, draw, fill=cyan!60!blue, minimum size=8mm] (a5) at (-3,4) {};

\draw[line width=2pt] (a1) -- (a2);
\draw[line width=2pt] (a1) -- (a3);
\draw[line width=2pt] (a2) -- (a5);

\draw[->, red, line width=2pt] (a4) -- (a0) node[midway, above, text=red, fill=none, draw=none] {\Large \,\,\,$\bm{\sigma}$};

\node (b1') at (1.5,0) {};
\node[circle, draw, fill=cyan!60!blue, minimum size=8mm] (b1) at (2,0) {};
\node[circle, draw, fill=cyan!60!blue, minimum size=8mm] (b2) at (3.5,1.5) {};
\node[circle, draw, fill=cyan!60!blue, minimum size=8mm] (b3) at (0,-3) {};
\node[circle, draw, fill=red, minimum size=8mm] (b4) at (3.5,-1.5) {}; 
\node (b0) at (3.5,0) {}; 

\draw[line width=2pt] (b1) -- (b2);
\draw[line width=2pt] (b4) -- (b3);
\draw[line width=2pt] (a4) -- (b1);

\draw[->, red, line width=2pt] (b0) -- (b1) node[midway, above , text=red, fill=none, draw=none] {\Large \,\,\,\,$\bm{\beta_e}$};

\draw[->, red, line width=2pt] (a4) -- (b1') node[midway, above, text=red, fill=none, draw=none] {\Large \!\!\!$\bm{\beta_d}$};

\end{tikzpicture}
    \caption{An illustration of the infection mechanisms in our stochastic model of disease spread. Infected nodes (red) infect susceptible nodes (blue) at a rate $\beta_d$ via edges. Infected nodes infect uncontaminated hyperedges (gray) that include them at a rate $\sigma$. Contaminated hyperedges (light orange) infect susceptible nodes that are attached to them at a rate $\beta_e$.}
    \label{fig:infection}
\end{figure}

Our disease-spread model is a continuous-time, discrete-state Markov process with state vector $(\bm{X}(t),\bm{Y}(t))$, where $\bm X(t)$ encodes the states of the nodes and $\bm Y(t)$ encodes the states of the hyperedges in $E_e$ (i.e., the hyperedges of size at least 3). For each node $i\in V$, the state is $X_i(t) = 1$ if node $i$ is \textit{infected} at time $t$ and $X_i(t) = 0$ if it is \textit{susceptible} at time $t$. For each hyperedge $\ell \in E_e$, the state is $Y_\ell(t) = 1$ if hyperedge $\ell$ is \textit{contaminated} (i.e., it has a high concentration of infected aerosols) at time $t$ and $Y_\ell(t) = 0$ if it is \textit{uncontaminated} (i.e., it has a low concentration of infected aerosols).

We now describe the update rules for the states of the nodes and the states of the hyperedges in $E_e$. 
\begin{enumerate}
    \item A susceptible node $i\in V$ becomes infected at rate
    \begin{equation}\label{eq:node_infected}
  	  \beta_d \sum_{h\in E_d} \I_{ih}\bigg(\sum_{j\neq i} X_{j} \I_{jh}\bigg) + \beta_e \sum_{h\in E_e} \I_{ih} Y_h \,,
    \end{equation}
    where $\beta_d$ and $\beta_e$ are the positive constants that are associated with the two transmission modes.  
    \item An infected node $i\in V$ becomes susceptible at node-recovery rate $\gamma$. 
    \item An uncontaminated hyperedge $\ell\in E_e$ becomes contaminated at rate
    \begin{equation}\label{eq:h_contaminated}
	    \sigma  g\bigg(\sum_{j\in V} X_{j} \I_{j\ell}\bigg) \,,
    \end{equation}
    where $\sigma$ is a constant that encodes environmental factors (such as temperature, humidity, and ultraviolet radiation) that affect aerosol transmission~\cite{wang2021airborne} and the function $g$ specifies how the
    infected individuals in an environment contribute to the environment's virus-transmission risk.
    For convenience, we assume that $g$ is a sigmoid function with $g(0) = 0$, $g'(x) > 0$, and $g''(x) < 0$ for $x\geq 0$. 
    We use the function $g(x) = \arctan(x)$.

    \item A contaminated hyperedge $\ell\in E_e$ becomes uncontaminated at hyperedge-recovery rate $\delta$. This recovery rate is affected by environmental factors such as airflow direction, ventilation, and air filtration and disinfection~\cite{wang2021airborne}. 
\end{enumerate}

We summarize the model parameters and their definitions in Table \ref{table:parameter}, and we illustrate the infection mechanisms of our disease-spread model in Figure \ref{fig:infection}. For most of our paper (except for Section \ref{sec:eff_delta}), we assume that the constants $\beta_d$, $\beta_e$, and $\gamma$ are the same for all nodes and that the constants $\sigma$ and $\delta$ are the same for all hyperedges in $E_e$. In Section \ref{sec:eff_delta}, we briefly use numerical simulations to examine the effects of nonuniformity in the hyperedge-recovery rate $\delta$. Future work can explore the impact of the choice of the function $g$ and of heterogeneities in the model parameters.

\begin{table}[h]
\begin{tabular}{|l|l|}
\hline
Parameter               & Definition                                                \\ \hline
$\beta_d$ & Infection rate of individuals on other individuals        \\ \hline
$\beta_e$ & Infection rate of environments on individuals             \\ \hline
$\sigma$   & Contamination rate of individuals on environments             \\ \hline
$\gamma$   & Recovery rate of individuals                              \\ \hline
$\delta$   & Recovery rate of environments                             \\ \hline
$g$                       & Function that describes the contribution of infected individuals to \\ & the contamination of an environment \\ \hline
\end{tabular}
\caption{The parameters in our disease-spread model on hypergraphs.}\label{table:parameter}
\end{table}


\subsection{Random-hypergraph models}\label{sec:hypergraph}

{In this subsection, we describe the three random-hypergraph models that we employ: complete $(2,s)$-uniform hypergraphs, {$(k_d,k_e)$}-regular $(2,s)$-uniform hypergraphs, and Erd\H{o}s--R\'enyi (ER) hypergraphs. In Section \ref{sec:meanfield}, we derive mean-field approximations of our stochastic disease-spread model \eqref{eq:node_infected}--\eqref{eq:h_contaminated} for complete $(2,s)$-uniform hypergraphs and regular $(2,s)$-uniform hypergraphs. 
In Section \ref{sec:simulation}, we perform numerical simulations of both our stochastic model and a mean-field approximation of it} on regular $(2,s)$-uniform hypergraphs and ER hypergraphs.


\subsubsection{Complete $(2,s)$-uniform hypergraphs}

In a complete $(2,s)$-uniform hypergraph, all hyperedges in $E_e$ have size $s$, all pairs of nodes have edges between them, and
every subset of $s$ nodes has a hyperedge between them.  


\subsubsection{{$(k_d,k_e)$}-regular $(2,s)$-uniform hypergraphs}

In a {$(k_d,k_e)$}-regular $(2,s)$-uniform hypergraph, all hyperedges in $E_e$ have size $s$ and each node is attached to $k_d$ edges and to $k_e$ hyperedges. For brevity, for the rest of our paper, we refer to these hypergraphs as 
``regular $(2,s)$-uniform hypergraphs".

We now describe how we generate regular $(2,s)$-uniform hypergraphs with $N = |V|$ nodes. We attach each node $i\in V$ to
 $k_d$ ``stubs" (i.e., loose ends)
 for the edges and to $k_e$ stubs for the hyperedges. We form an edge by selecting 2 stubs uniformly at random, and we form a hyperedge by selecting $s$ stubs uniformly at random.
 We repeat this process until we have associated all stubs with nodes. We require that both $Nk_d/2$ and $Nk_e/s$ are integers. It is possible for a regular $(2,s)$-uniform hypergraph to have
 self-edges, self-hyperedges, multi-edges, and/or multi-hyperedges. We choose to keep any such edges and hyperedges.\footnote{\label{note1}In practice, none of our examined networks have any self-edges, self-hyperedges, multi-edges, or multi-hyperedges.} See \cite{fosdick2018configuring} for a detailed discussion of the consequences of such a choice in the context of ordinary graphs.
 
This type of hypergraph is a special case of Chodrow's stub-labeled configuration model of hypergraphs \cite{chodrow2020configuration}. In Chodrow's hypergraph configuration model, one specifies the number of nodes, the number of hyperedges, a hyperdegree sequence $\bm k$, and a hyperedge-size sequence $\bm s$. 
{In Section \ref{sec:simulation}, we simulate our stochastic disease-spread model \eqref{eq:node_infected}--\eqref{eq:h_contaminated} on regular $(2,s)$-uniform hypergraphs.} We also simulate our stochastic disease-spread model \eqref{eq:node_infected}--\eqref{eq:h_contaminated} on configuration-model hypergraphs with power-law hyperdegree sequences but do not observe any significant qualitative impact of 
the hyperdegee-sequence choice
 on the disease dynamics, so we omit these simulations from the present paper. 
Future work can further examine our disease-spread model on
configuration-model hypergraphs and their generalizations.


\subsubsection{ER hypergraphs}

We also consider a random-hypergraph model that is a natural extension of $G(N,M)$ ER graphs. 
We specify both the number of hyperedges and the hyperedge sizes. 
We denote a $G(N,M)$ ER hypergraph by $G(N,M_d,\bm M_e, \bm s)$, where the vectors $\bm M_e$ and $\bm s$ specify the number of hyperedges of each size. 
There are $M_{e,i}$ hyperedges of size $s_i$ for $i \in \{1,\ldots,|\bm s| \}$. For each hyperedge of size $s_i$, we uniformly-randomly select $s_i$ nodes from the set of $N$ nodes.
The hypergraphs that we construct in this manner can have self-edges, self-hyperedges, multi-edges, and/or multi-hyperedges; we choose to keep them.\footnote{For this model as well, none of our networks have any self-edges, self-hyperedges, multi-edges, or multi-hyperedges in practice.} 


\section{Mean-field approximations}\label{sec:meanfield}

In this section, we derive mean-field approximations of the stochastic disease-spread model \eqref{eq:node_infected}--\eqref{eq:h_contaminated} for both complete $(2,s)$-uniform hypergraphs (see Section \ref{sec:meanfield_complete}) and regular 
$(2,s)$-uniform hypergraphs (see Section \ref{sec:meanfield_regular}). We also briefly discuss some practical implications of the latter, as we use it to guide many simulations in Section \ref{sec:simulation}.


\subsection{Mean-field approximation for complete $(2,s)$-uniform hypergraphs}\label{sec:meanfield_complete}

We develop a mean-field approximation for complete $(2,s)$-uniform hypergraphs and establish an extinction threshold (i.e., a condition for the local stability of the disease-free equilibrium) of the disease in the mean-field model. 
We use an individual-level approach~\cite{pastor2015epidemic}.

Consider the expected value of $X_i(t)$ for nodes $i\in V$ and $Y_\ell(t)$ for hyperedges $\ell\in E_e$:
\begin{align*}
    x_i(t) &:= \E(X_i(t)) = \PP(X_i(t) = 1) \,,\\
    y_\ell(t) &:= \E(Y_\ell(t)) = \PP(Y_\ell(t) = 1)\,.
\end{align*}
For a node $i\in V$, we have
\begin{align*}
    \dot{x}_i &= \E\bigg[ \beta_d \sum_{h\in E_d} \I_{ih} \big(\sum_{j\neq i} X_{j} \I_{jh}\big)+ \beta_e \sum_{h\in E_e} \I_{ih} Y_h\bigg]\PP(X_i(t)=0) -\gamma \PP(X_i(t)=1)\\
    	&= \bigg(\beta_d \sum_{h\in E_d} \I_{ih} \E\big[\sum_{j\neq i} X_{j} \I_{jh}\big]+ \beta_e \sum_{h\in E_e} \I_{ih} \E [Y_h]\bigg)(1-x_i) - \gamma x_i\\
   	 &= \bigg(\beta_d \sum_{h\in E_d} \I_{ih} \sum_{j\neq i} x_j \I_{jh} + \beta_e \sum_{h\in E_e} \I_{ih} y_h\bigg)(1-x_i) - \gamma x_i\,.
\end{align*}
Similarly, for a hyperedge $\ell\in E_e$, we have
\begin{align*}
    \dot{y}_\ell = \sigma \E\bigg[g\big(\sum_{j\in V} X_i\I_{i\ell}\big)\bigg](1 - y_\ell) - \delta y_\ell\,.
\end{align*}
We introduce a further approximation by interchanging the operation of the function
$g$ and taking the expectation to obtain a deterministic mean-field ordinary-differential-equation (ODE) system
\begin{align}\label{eqn:xiyl}
    \dot{x}_i &= F_i(\bm x,\bm y) := \bigg(\beta_d \sum_{h\in E_d}\I_{ih} \sum_{j\neq i} x_j \I_{jh} + \beta_e \sum_{h\in E_e}\I_{ih} y_h\bigg)(1 - x_i) - \gamma x_i \,, \nonumber\\
    \dot{y}_\ell &= G_\ell(\bm x,\bm y) := \sigma g\big(\sum_{j\in V} x_j\I_{j\ell}\big)(1 - y_\ell) - \delta y_\ell \,.
\end{align}

We now provide a sufficient condition for the disease-free equilibrium (DFE) $\bm 0 \in \mathbb{R}^{N + M_e}$ of \eqref{eqn:xiyl} to be locally asymptotically stable. 
\begin{theorem}
The DFE of \eqref{eqn:xiyl} is locally asymptotically stable if 
    $R^{\mathrm{c}}_0 < 1$ and is unstable if $R^{\mathrm{c}}_0 > 1$, where 
    \begin{equation}\label{eq:R0c}
		R^{\mathrm{c}}_0 = \frac{\beta_d(N  - 1)}{\gamma} + \frac{\beta_e \sigma g'(0){N - 1 \choose s - 1}s}{\gamma\delta}\,.  
\end{equation}
    
\end{theorem}

\begin{proof}
    Because $g(0) = 0$, the point $\bm 0 \in \mathbb{R}^{N + M_e}$ is an equilibrium of the system \eqref{eqn:xiyl}. It suffices to show that every eigenvalue of the Jacobian matrix $J(\bm 0)$ has a negative real part. 
    
    We compute 
 \begin{align*}
    \frac{\partial F_i}{\partial x_j} &= \begin{cases}
    \beta_d \big(\sum_{h\in E_d}\I_{ih} \I_{jh}\big)(1 - x_i) \quad &\text{if}\quad j\neq i\\
    - \big(\beta_d\sum_{h\in E_d}\I_{ih} \sum_{j\in V}x_j  \I_{jh} + \beta_e \sum_{h\in E_e} \I_{ih} y_h\big) - \gamma \quad &\text{if} \quad j=i\,,
    \end{cases}\\
    \frac{\partial F_i}{\partial y_h} &= \beta_e \I_{ih}(1-x_i)\,,\\
    \frac{\partial G_\ell}{\partial x_j} &= \sigma g'\big(\sum_{j\in V} x_j\I_{j\ell}\big)\I_{j\ell}(1 - y_\ell)\,,\\
    \frac{\partial G_\ell}{\partial y_h} &= \begin{cases}
        0\quad &\text{if}\quad m\neq \ell\\
        -\delta \quad&\text{if}\quad m = \ell \,.
    \end{cases}
    \end{align*}
    The Jacobian evaluated at the DFE
    is 
    \[
   	 J(\bm 0) = \begin{bmatrix}
   	     \beta_d \widetilde{\W}^{d} - \gamma I_{N\times N} &\beta_e \I^e\\
   	     g'(0)\sigma (\I^e)^T &-\delta I_{M_e\times M_e} 
   	 \end{bmatrix}\,,
    \]
    where $I_{N\times N}$ and $I_{M_e\times M_e}$ are identity matrices of sizes $N$ and $ M_e$, respectively, and $\widetilde{\W}^{d}$ is the matrix that we obtain from $\W^d$ after replacing its diagonal entries by $0$.

    The associated characteristic polynomial is
    \begin{align*}
  	  p(\lambda)&= \det\begin{bmatrix}
  		      A  &B\\
        C &D 
  		  \end{bmatrix}:=\det\begin{bmatrix}
        	\beta_d\widetilde{\W}^{d} - (\gamma +\lambda) I_{N\times N}  &\beta_e \I_e\\
        	\sigma g'(0) (\I^e)^T &-(\delta +\lambda) I_{M_e\times M_e} 
  	  \end{bmatrix}  = \det(D)\det(A-BD^{-1}C)\,.
    \end{align*}
   Note that $\det(D) = (-1)^{M_e} (\lambda+\delta)^{M_e}$ and $D^{-1} = \frac{-1}{\delta+\lambda}I_{M_e\times M_e}$. Therefore,
    \[
	    p(\lambda) = (-1)^{M_e} (\lambda+\delta)^{M_e} q(\lambda)\,, 
    \]
    where 
    \[
  	  q(\lambda)= \det(\beta_d\widetilde{\W}^{d} - (\gamma +\lambda) I_{N\times N} +\frac{\beta_e \sigma g'(0)}{\delta+\lambda}\W^e)\,.
    \]

       It now suffices to show that every solution of $q(\lambda) = 0$ has a negative real part. The hypergraph on which we {approximate the model \eqref{eq:node_infected}--\eqref{eq:h_contaminated}} 
    is a complete $(2,s)$-uniform hypergraph, so we directly compute 
    \[
	    \widetilde{\W}^{d} = \bm 1\bm 1^T - I_{N\times N}
    \]
    and 
    \[
  	  \W^e = {N - 2\choose s - 2}\bm 1\bm 1^T + \bigg[{N - 1\choose s - 1} - {N - 2\choose s - 2} \bigg]I_{N\times N} = {N - 2\choose s - 2}\bm 1\bm 1^T + {N - 2\choose s - 2}\frac{N - s}{s - 1}I_{N\times N}\,.
    \]
    We thereby obtain
    \[
	    q(\lambda) = \det(I_{N\times N} f_1(\lambda) + \bm 1\bm 1^T f_2(\lambda))\,,
    \]
    where 
  \begin{align*}
	    f_1(\lambda) &= {N - 2\choose s - 2}\frac{N - s}{s - 1}\frac{\beta_e \sigma g'(0)}{\delta + \lambda} - \beta_d - (\gamma + \lambda) \,, \\
	    f_2(\lambda) &= \beta_d + {N - 2\choose s - 2}\frac{\beta_e \sigma g'(0)}{\delta + \lambda}\,.
    \end{align*}
    Using the matrix determinant lemma~\cite{harville1998matrix}, we write
    \begin{align*}
	     q(\lambda) &= \det\bigg(f_1(\lambda)(I_{N\times N} + \frac{f_2(\lambda)}{f_1(\lambda)}\bm 1\bm 1^T \bigg) \\
	     &= f_1(\lambda)^N\bigg(1 + \frac{f_2(\lambda)}{f_1\lambda)} \bm 1^T\bm 1\bigg) \\
	     &= f_1(\lambda)^N \bigg(1 + \frac{Nf_2(\lambda)}{f_1(\lambda)}\bigg)\,.
    \end{align*}
Therefore, the solutions of $q(\lambda) = 0$ satisfy either $f_1(\lambda) = 0$ or $f_1(\lambda) + Nf_2(\lambda) = 0$. 

Direct calculations show that the solutions of $f_1(\lambda) = 0$ are the solutions of the quadratic equation $\lambda^2 + a_1\lambda + a_2 = 0$, with
\begin{align*}
	a_1 &= \gamma + \delta + \beta_d \,, \\
	a_2 &= \gamma\delta + \beta_d\delta - \beta_e \sigma g'(0){N - 2\choose s - 2}\frac{N - s}{s - 1}\,.
\end{align*}
Similarly, direct calculations show that the solutions of $f_1(\lambda) + Nf_2(\lambda) = 0$ are solutions of the quadratic equation $\lambda^2 + b_1\lambda + b_2 = 0$, with
\begin{align*}
	b_1 &= \gamma + \delta - \beta_d(N - 1) \,, \\
	 b_2 &= \gamma\delta - \beta_d(N - 1)\delta - \beta_e \sigma g'(0){N - 1\choose s - 1}s = \gamma\delta(1 - R_0)\,.
\end{align*}
We consider two cases.

\medskip

\textbf{Case 1:} Suppose that $R^{\mathrm{c}}_0 < 1$.
We  have
\[
	\beta_e \sigma g'(0){N - 2\choose s - 2}\frac{N - s}{s - 1} < \beta_e\sigma g'(0){N - 1\choose s - 1}s < \gamma\delta\,,
\]
where the second inequality follows from the condition $R^{\mathrm{c}}_0 < 1$.Therefore, both $a_1$ and $a_2$ are positive, which implies that the solutions of $f_1(\lambda) = 0$ all have negative real parts.

The condition $R^{\mathrm{c}}_0 < 1$ implies that $b_2 > 0$ and $\beta_d(N - 1) < \gamma$. The latter implies that $b_1 > 0$. Consequently, all solutions of $f_1(\lambda) + Nf_2(\lambda) = 0$ must have negative real parts. 
Therefore, the DFE of \eqref{eqn:xiyl} is locally asymptotically stable.

\medskip

\textbf{Case 2:} Suppose that $R^{\mathrm{c}}_0 > 1$. This condition implies directly that $b_2 < 0$. Therefore, the quadratic equation $\lambda^2 + b_1\lambda + b_2 = 0$ must have a positive solution. Consequently, the DFE of \eqref{eqn:xiyl} is unstable.
\end{proof}

We end this subsection with a brief remark about the expression for $R_0^{\mathrm{c}}$ in \eqref{eq:R0c}. 

\begin{remark}\label{rem:meanfield_complete}
The expression \eqref{eq:R0c} is a sum of contributions of the two disease-transmission modes. The only parameter that affects both terms is the 
node-recovery rate (i.e., individual recovery rate) $\gamma$. When the contribution of one mode exceeds $1$,
it is not possible to adjust the other mode's parameters (except for $\gamma$) to reduce $R_0^{\mathrm{c}}$ below $1$. 
\end{remark}


\subsection{Mean-field approximation for regular $(2,s)$-uniform hypergraphs}\label{sec:meanfield_regular}

In this subsection, we give a  mean-field approximation for regular $(2,s)$-uniform hypergraphs. In such a hypergraph, recall that all hyperedges are of size $s$ and that each node is attached to $k_d$ edges and $k_e$ hyperedges.

Using a degree-based mean-field approach~\cite{pastor2015epidemic}, we assume that all nodes behave in the same way and that all the hyperedges in $E_e$ behave in the same way. 
In particular, we approximate the dynamics of the stochastic model \eqref{eq:node_infected}--\eqref{eq:h_contaminated} by the deterministic ODE system
\begin{align}\label{eqn:meanfield}
    \dot{x} &= (\beta_dk_dx + \beta_ek_ey)(1 - x) - \gamma x\nonumber\,,\\
    \dot{y} &= \sigma g(sx)(1 - y) - \delta y\,,
\end{align}
where $x(t)$ is the proportion of infected nodes at time $t$ and $y(t)$ is the proportion of contaminated hyperedges at time $t$.

We will show that the basic reproduction number for the mean-field model \eqref{eqn:meanfield} has a similar form as the basic reproduction number $R^{\mathrm{c}}_0$
in Section \ref{sec:meanfield_complete}.

Specifically, we show that the basic reproduction number for the model \eqref{eqn:meanfield} is
\begin{equation}\label{eqn:R0}
    R^{\mathrm{ru}}_0 = \frac{\beta_dk_d}{\gamma} + \frac{\beta_e\sigma g'(0)k_e s}{\gamma\delta}\,.
\end{equation}

Before proving that \eqref{eqn:R0} is indeed the basic reproduction number for the model \eqref{eqn:meanfield}, we discuss several potential implications of the expression for $R^{\mathrm{ru}}_0$.
First, as with \eqref{eq:R0c}, the expression
\eqref{eqn:R0} is a sum of contributions of the two disease-transmission modes.

The only parameter that affects both terms is the 
node-recovery rate (i.e., individual recovery rate) $\gamma$. If individuals 
recover quickly, then they are also less likely to contaminate environments,
which in turn are less likely to infect their occupants. Another implication is that if the contribution of the dyadic transmission mode (i.e., infection through social contacts) satisfies $\frac{\beta_dk_d}{\gamma} > 1$, then it is not possible to control the environmental factors  (i.e., the hyperedge parameters $\beta_e$, $\sigma$, $k_e$, $s$, and $\delta$) to reduce the basic reproduction number $R^{\mathrm{ru}}_0$ below $1$. 
Finally, if $\frac{\beta_dk_d}{\gamma} < 1$, then in addition to control measures that reduce infection 
rates (through the parameters $\beta_d$, $\beta_e$, and $\sigma$) or social contacts (through the parameters $k_d$, $k_e$, and $s$), it is also true that increasing the hyperedge-recovery rate (i.e., environmental recovery rate) $\delta$ can help reduce $R^{\mathrm{ru}}_0$. In a real-life context, one can do this by improving {the ventilation and air filtration of indoor buildings~\cite{wang2021airborne, morawska2021paradigm}.

\begin{theorem}\label{thm:extinction}
When $R^{\mathrm{ru}}_0 < 1$, the DFE
of the mean-field model \eqref{eqn:meanfield} is locally asymptotically stable.
\end{theorem}

\begin{proof}
Because $g(0) = 0$, the DFE $\bm 0\in \mathbb{R}^{N + M_e}$ is an equilibrium of the system \eqref{eqn:meanfield}. Linearizing around this equilibrium yields the Jacobian matrix
\[
	J(\bm 0)=\begin{bmatrix}
	    \beta_d k_d -\gamma & \beta_ek_e\\
	    \sigma s g'(0) & -\delta\
	\end{bmatrix}\,.
\]
We calculate the determinant of $J(\bm 0)$ to obtain $\det(J(\bm 0)) = \gamma\delta - \beta_dk_d\delta - \beta_e\sigma g'(0)k_e s = \gamma\delta(1 - R^{\mathrm{ru}}_0) > 0$. The trace of $J(\bm 0)$ is $\text{Tr}(J(\bm 0)) = \beta_dk_d - \gamma - \delta < \beta_dk_d - \gamma < 0$ because $\frac{\beta_dk_d}{\gamma} < R^{\mathrm{ru}}_0 < 1$. Therefore, the eigenvalues of $J(\bm 0)$ must have negative real parts, which implies that the DFE is locally asymptotically stable. 
\end{proof}

\begin{lemma}\label{lem:existence}
When $R^{\mathrm{ru}}_0 > 1$, the mean-field model \eqref{eqn:meanfield} has a unique positive equilibrium.
\end{lemma}

\begin{proof}
    Setting $\dot{y} = 0$ in \eqref{eqn:meanfield} and solving for $y$ yields
    \begin{equation}\label{eqn:y}
   	 y = \frac{\sigma g(sx)}{\delta + \sigma g(sx)}\,.
    \end{equation}
    Setting $\dot{x} = 0$ and using \eqref{eqn:y} implies that the equilibrium value of $x$
    must satisfy
    \[
	    \F(x) := \bigg[\beta_dk_d x + \beta_ek_e\frac{\sigma g(sx)}{\delta+\sigma g(sx)}\bigg](1 - x) - \gamma x = 0\,.
    \]
    We now show that $\F(x) = 0$ has a unique solution $x^*\in(0,1)$. For convenience, let 
    \begin{equation}\label{eq:G}
   	 \G(x) = \beta_dk_d x + \beta_ek_e\frac{\sigma g(sx)}{\delta + \sigma g(sx)}\,.
    \end{equation}
The first two derivatives of $\F$ are
    \begin{equation}\label{eq:F'}
   	 \F'(x) = \G'(x)(1 - x) - \G(x) - \gamma
    \end{equation}
    and
    \[
	    \F''(x) = \G''(x)(1 - x) - 2\G'(x)\,.
    \]
    Because $g'(sx) > 0$, we have
    \begin{equation}\label{eq:G'}
    	\G'(x) = \beta_dk_d + \beta_ek_e \sigma g'(sx) s \frac{\delta}{(\delta+\sigma g(sx))^2} > 0
    \end{equation}
    for any $x\geq 0$. Furthermore, 
    \[
	    \G''(x) = \beta_ek_e\sigma s\delta\bigg(\frac{g''(sx)s}{(\delta + \sigma g(sx))^2} - \frac{2g'(sx)^2s}{(\delta + \sigma g(sx))^3}\bigg) < 0
    \]
    for any $x\geq 0$ because $g'(sx) > 0$ and $g''(sx) < 0$. Therefore, $\F''(x) < 0$, which implies that $\F'(x)$ is decreasing. At $x = 0$, we have
    \[
	    \F'(0) = \G'(0) - G(0) - \gamma = \beta_dk_d + \frac{\beta_ek_e\sigma g'(0)}{\delta} - \gamma = \gamma(R^{\mathrm{ru}}_0 - 1) > 0
    \]
    and
    \[
	    \F'(1) = -\G(1) - \gamma < 0\,.
    \]
    Therefore, $\F'(x) = 0$ has a unique solution $c\in (0,1)$.  Furthermore, $\F'(x) > 0$ for $x\in(0,c)$ and $\F'(x) < 0$ for $x\in (c,1)$. Because $\F(0) = 0$ and $\F(1) = -\gamma < 0$, no solution of $\F(x) = 0$ lies in $(0,c)$ and there is a unique solution $x^*\in (c,1)$.

    Using \eqref{eqn:y}, we see that the equilibrium value of $y$ is
    \begin{equation}\label{eqn:y*}
    	y^* = \frac{\sigma g(sx^*)}{\delta + \sigma g(sx^*)} \in (0,1)\,,
    \end{equation}
    which implies that the mean-field model \eqref{eqn:meanfield} has a unique positive equilibrium $(x^*,y^*)$.
    
    \end{proof}

\begin{lemma}\label{lem:pos_localstability}
    When $R^{\mathrm{ru}}_0 > 1$, the unique positive equilibrium (see Lemma \ref{lem:existence}) is locally asymptotically stable.
\end{lemma}

\begin{proof}
    Linearizing around the positive equilibrium $(x^*,y^*)$ yields the Jacobian matrix
    \[
   	 J(x^*,y^*) = \begin{bmatrix}
        		\beta_dk_d(1 - x^*) - (\beta_dk_dx^* + \beta_ek_ey^*) - \gamma & \beta_ek_e(1 - x^*)\\
        		\sigma s g'(sx^*)(1 - y^*) & -\sigma g(sx^*) - \delta
    \end{bmatrix}\,.
    \]
    We observe that $x^*$ satisfies
    \[
    \gamma x^* = (\beta_dk_dx^* + \beta_ek_ey^*)(1 - x^*)  > \beta_dk_dx^* (1 - x^*) \,,
    \]
    which implies that $\gamma > \beta_dk_d(1 - x^*)$. Therefore,
    \[
	    \text{Tr}(J(x^*,y^*)) = \beta_dk_d(1 - x^*) - (\beta_dk_dx^* + \beta_ek_ey^*) - \gamma  - \sigma g(sx^*) - \delta < 0\,.
    \]
To see that that $\det(J(x^*,y^*)) > 0$, we calculate
    \begin{align*}
	   \det(J(x^*,y^*)) &= (\beta_dk_dx^* + \beta_ek_ey^* + \gamma)(\sigma g(sx^*)+\delta) - (1 - x^*)(\beta_dk_d(\sigma g(sx^*) + \delta) + \beta_e k_e\sigma s g'(sx^*)(1 - y^*))\\
		    &= (\sigma g(sx^*) + \delta)\bigg(\beta_dk_dx^* + \beta_ek_ey^* + \gamma - (1 - x^*)\bigg(\beta_dk_d + \beta_e k_e\sigma s g'(sx^*)\frac{\delta}{(\delta + \sigma g(sx^*))^2}\bigg)\bigg)\\
		    &= (\sigma g(sx^*) + \delta)(\beta_dk_dx^* + \beta_ek_ey^* + \gamma - (1 - x^*)\G'(x^*))\,,
    \end{align*}
    where the second equality follows from \eqref{eqn:y*} and {the third equality follows from the definition of $\G$ in \eqref{eq:G} and the calculation of $\G'$ in \eqref{eq:G'}.}
       By Lemma \ref{lem:existence}, we also have $\F'(x) < 0$ for $x\in(c,1)$. Because $x^*\in(c,1)$, we have $\F'(x^*) < 0$. Using \eqref{eq:F'} yields
    \[
   	 \G'(x^*)(1 - x^*) < \G(x^*) + \gamma = \beta_dk_dx^* + \beta_ek_e\frac{\sigma g(sx^*)}{\delta + \sigma g(sx^*)} + \gamma = \beta_dk_dx^* + \beta_ek_e y^* + \gamma\,.
    \]
    Therefore, $\det(J(x^*,y^*)) > 0$. 
    
Because $\text{Tr}(J(x^*,y^*) < 0$ and $\det(J(x^*,y^*)) > 0$, all eigenvalues of the matrix $J(x^*,y^*)$ have negative real parts, which implies that the positive equilibrium is locally asymptotically stable.
\end{proof}

\begin{theorem}\label{thm:endemic}
When $R^{\mathrm{ru}}_0 > 1$, the mean-field model \eqref{eqn:meanfield} has a unique positive equilibrium. Furthermore, the model's DFE is unstable and the positive equilibrium is locally asymptotically stable.
\end{theorem}

\begin{proof}
    The existence, uniqueness, and local stability of the positive equilibrium follow from Lemma \ref{lem:existence} and Lemma \ref{lem:pos_localstability}. From Theorem \ref{thm:extinction}, it follows that $\det(J(\bm 0)) = \gamma\delta(1 - R^{\mathrm{ru}}_0) < 0$, so $J(\bm 0)$ has a positive eigenvalue and a negative eigenvalue, which implies that the DFE is unstable.
\end{proof}

\begin{remark}\label{remark:size}
   We can express the basic reproduction number $R^{\mathrm{ru}}_0$ for \eqref{eqn:meanfield} in terms of the numbers of edges and hyperedges, rather than in terms of the degree and hyperdegree (which is what we did in \eqref{eqn:R0}).
    The number of edges is $M_d = {Nk_d}/{2}$, and the number of hyperedges is 
   $M_e = {Nk_e}/{s}$. Therefore,
   \begin{equation}\label{eqn:R0_alt}
  	 R^{\mathrm{ru}}_0 = \frac{2\beta_dM_d}{N\gamma} + \frac{\beta_e\sigma g'(0)M_e s^2}{N\gamma\delta}\,.
   \end{equation}
   The expression \eqref{eqn:R0_alt} for $R^{\mathrm{ru}}_0$ implies that 
   the hyperedge size $s$ has
   a strong influence on $R^{\mathrm{ru}}_0$. 
   For example, if one doubles the hyperedge size, then to keep $R^{\mathrm{ru}}_0$ unchanged, one needs to quadruple the recovery rates $\gamma$ or $\delta$ or to decrease the infection rates $\beta_e$ or $\sigma$ by a factor of four.
    Accordingly, reducing the hyperedge size $s$ (e.g., by limiting the size of indoor events) is 
    an effective way to decrease $R^{\mathrm{ru}}_0$ if the dyadic transmission
    satisfies  
   $\frac{2\beta_dM_d}{N\gamma} < 1$. 
\end{remark}


\section{Simulations of our stochastic model \eqref{eq:node_infected}--\eqref{eq:h_contaminated} and the mean-field model \eqref{eqn:meanfield}}\label{sec:simulation}

In this section, we present the results of  
simulations of the stochastic model \eqref{eq:node_infected}--\eqref{eq:h_contaminated} to gain insight into the accuracy of the mean-field approximation \eqref{eqn:meanfield}.\footnote{We do not include comparisons between
simulations of 
\eqref{eq:node_infected}--\eqref{eq:h_contaminated} and the mean-field approximation \eqref{eqn:xiyl} 
because complete $(2,s)$-uniform hypergraphs have a very large number of hyperedges, and it is thus computationally costly to simulate the stochastic model \eqref{eq:node_infected}--\eqref{eq:h_contaminated}.}
We also use simulations to investigate
the effects of various network parameters and other model parameters on the disease dynamics, especially when we incorporate heterogeneity into hypergraphs.

In Section \ref{sec:algo}, we specify our simulation algorithm. In Section \ref{sec:compare}, we compare simulations of the stochastic model \eqref{eq:node_infected}--\eqref{eq:h_contaminated} to simulations of the mean-field approximation \eqref{eqn:meanfield}. In Sections \ref{sec:eff_size} and \ref{sec:eff_delta}, we discuss the effects on disease dynamics} of hypergraph-size distribution and nonuniform hyperedge-recovery rates, respectively.


\subsection{Simulation algorithm}\label{sec:algo}

We use a common individual-level stochastic-simulation algorithm~\cite{graham2013stochastic}. At 
time $t$, we generate a uniformly random vector $\bm r\in[0,1]^{N + M_e}$. 
The algorithm, which runs through all nodes from $i = 1$ to $i = N$ and all hyperedges in $E_e$ from $\ell = 1$ to $\ell = M_e$, proceeds as follows.
\begin{enumerate}
    \item If node $i$ is susceptible, it becomes infected at time $t + \Delta t$ if
    \[
	    r_i < 1- \exp\bigg(-\bigg(\beta_d \sum_{h\in E_d} \I_{ih} \big(\sum_{j\neq i} X_{j} \I_{jh}\big) + \beta_e \sum_{h\in E_e} \I_{ih} Y_h\bigg)\Delta t\bigg)\,.
    \]
    \item If node $i$ is infected, it becomes susceptible at time $t + \Delta t$ if
    \[
	    r_i <1 - \exp(-\gamma\Delta t)\,.
    \]
    \item If hyperedge $\ell$ is uncontaminated, it becomes contaminated at time $t + \Delta t$ if 
    \[
	    r_{N + \ell} < 1 - \exp\bigg(-\sigma  g\big(\sum_{j\in V} X_{j} \I_{j\ell}\big)\Delta t\bigg)\,.
    \]
    In our simulations, we use $g(x) = \arctan(x)$, which satisfies the assumptions in Section \ref{sec:infection}.
    \item If hyperedge $\ell$ is contaminated, it becomes uncontaminated at time $t + \Delta t$ if
    \[
	    r_{N + \ell} < 1 - \exp(-\delta\Delta t)\,.
    \]
\end{enumerate}

We use the step size $\Delta t = 0.1$, and we run our simulations for either 400 or 800 time steps. (We use the latter in situations when it takes longer for the system \eqref{eq:node_infected}--\eqref{eq:h_contaminated}  to converge.)
We specify the initial fraction $p_0$ of infected nodes, which we choose uniformly at random.
In all of our simulations, there are initially no contaminated hyperedges. For each choice of initial condition, parameter values, and hypergraph, we run 10 simulations on the same hypergraph and average our results to smooth out the simulation results. Our code is available at \url{https://github.com/TungDaoNguyen/SIS_Hypergraph}.


\subsection{Comparison between the the stochastic model \eqref{eq:node_infected}--\eqref{eq:h_contaminated} and the mean-field model \eqref{eqn:meanfield}} 
\label{sec:compare}

We first compare 
simulations of the individual-level stochastic model \eqref{eq:node_infected}--\eqref{eq:h_contaminated} and the mean-field approximation \eqref{eqn:meanfield} for regular $(2,s)$-uniform hypergraphs. 

For a set of simulations, we generate a single regular $(2,s)$-uniform hypergraph with $N = 100$ nodes, hyperedge size $s = 4$, degree $k_d = 20$, and hyperdegree $k_e = 4$. This hypergraph has 1000 edges and 100 hyperedges. 
We compare the proportion $\bar{X}(t)$ of infected nodes 
and the proportion $\bar{Y}(t)$ of contaminated hyperedges 
in the stochastic model \eqref{eq:node_infected}--\eqref{eq:h_contaminated} with the corresponding proportions in the mean-field approximation \eqref{eqn:meanfield}.
We present results from one set of simulations, and we observe the same qualitative dynamics for other sets of simulations (using several different regular $(2,s)$-uniform hypergraphs).

In Figure \ref{fig:meanfield_regular}, we illustrate that the mean-field model provides a good approximation of the stochastic simulations in two scenarios: (a) when we choose parameters
so that $R^{\mathrm{ru}}_0$ in \eqref{eqn:R0} is well above 1
(the value of $R^{\mathrm{ru}}_0$ is 3.60); and (b) when we choose parameters
so that $R^{\mathrm{ru}}_0$ in \eqref{eqn:R0} is reasonably below 1
(the value of $R^{\mathrm{ru}}_0$ is 0.80). 
For convenience, we refer to the former scenario as an ``endemic scenario" and the latter scenario as a ``disease-free scenario".

\begin{figure}
    \centering
    \begin{subfigure}[t]{0.45\textwidth}
        \centering
        \includegraphics[width=\linewidth]{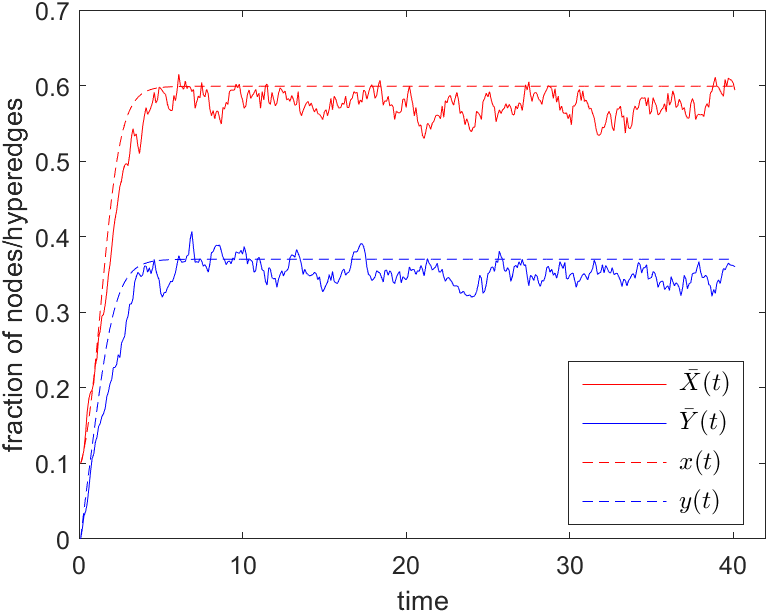} 
        \caption{Endemic scenario: individual infection rate $\beta_d = 0.1$, environmental infection rate $\beta_e = 0.2$, environmental contamination rate $\sigma = 0.5$, and initial infected-node proportion $p_0 = 0.1$} 
    \end{subfigure}
    \hfill
    \begin{subfigure}[t]{0.45\textwidth}
        \centering
        \includegraphics[width=\linewidth]{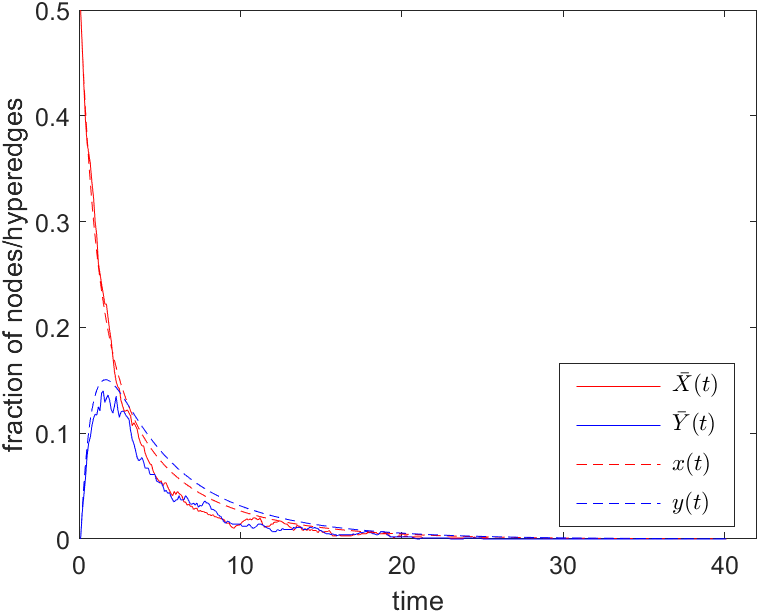} 
        \caption{Disease-free scenario: individual infection rate $\beta_d = 0.03$, environmental infection rate $\beta_e = 0.05$, environmental contamination rate $\sigma = 0.25$, and initial infected-node proportion $p_0 = 0.5$} 
    \end{subfigure}
    \caption{Proportions of infected nodes (red) and contaminated hyperedges (blue), averaged over 10 simulations, in the individual-level stochastic model \eqref{eq:node_infected}--\eqref{eq:h_contaminated} (solid curves) and proportions of infected nodes (red) and contaminated hyperedges (blue) in the mean-field model \eqref{eqn:meanfield} (dashed curves) on a regular $(2,s)$-uniform hypergraph. 
    The recovery rates of the nodes and hyperedges are $\gamma = \delta = 1$. 
    In (a), we show results for an endemic scenario. In (b), we show results for a disease-free scenario.}\label{fig:meanfield_regular}
\end{figure}

Although the mean-field approximation \eqref{eqn:meanfield} is designed for regular $(2,s)$-uniform hypergraphs, we also test it on simulations of the model \eqref{eq:node_infected}--\eqref{eq:h_contaminated} on an ER
hypergraph. 
We generate an ER hypergraph with $N = 100$ nodes, 1000 edges, and 100 hyperedges of size $4$. This hypergraph has a mean degree of $\bar{k}_d = 20$ and a mean hyperdegree of $\bar{k}_e = 4$. We use these values in the mean-field model instead of $k_d$ and $k_e$. In Figure \ref{fig:meanfield_ER}, we observe that the mean-field model is just as successful on the ER hypergraph as it was on the regular $(2,s)$-uniform hypergraph.
We present results from one set of simulations, and we observe the same qualitative dynamics for other sets of simulations (using several different ER hypergraphs).

Recall that the expression \eqref{eqn:R0} for $R^{\mathrm{ru}}_0$ in the mean-field model \eqref{eqn:meanfield} has separate contributions from the dyadic and polyadic
transmission modes. 
To examine this feature, we perform further simulations of the individual-level stochastic model \eqref{eq:node_infected}--\eqref{eq:h_contaminated} and the mean-field model \eqref{eqn:meanfield} on the ER hypergraph above. 
In Figure \ref{fig:R0_ER}(a), we consider parameters in which
the dyadic transmission mode (i.e., transmission via edges) gives the dominant contribution to $R^{\mathrm{ru}}_0$.
We observe that halving the individual infection rate $\beta_d$ drives a disease to extinction quickly, whereas halving the environmental infection rate $\beta_e$ or environmental contamination rate $\sigma$ do not significantly impact the disease dynamics. In Figure \ref{fig:R0_ER}(b), we consider parameters in which the polyadic transmission mode (i.e., transmission via hyperedges) gives the dominant contribution to $R^{\mathrm{ru}}_0$. In this case, halving $\beta_e$ or $\sigma$ quickly drives a disease to extinction. Halving $\beta_d$ still leads to disease extinction, but the disease is eradicated
more slowly than halving $\beta_e$ or $\sigma$.
From these simulations, we conclude that we are able to distinguish between the effects of the two transmission modes in the individual-level stochastic model \eqref{eq:node_infected}--\eqref{eq:h_contaminated} but that the separation between their effects is not as strong as suggested by the mean-field model \eqref{eqn:meanfield}.

\begin{figure}
    \centering
    \begin{subfigure}[t]{0.45\textwidth}
        \centering
        \includegraphics[width=\linewidth]{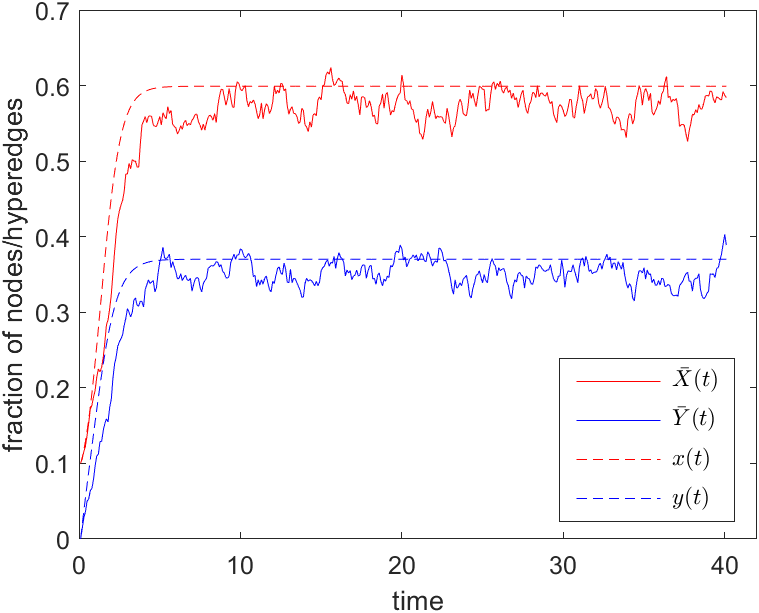} 
        \caption{Endemic scenario: individual infection rate $\beta_d = 0.1$, environmental infection rate $\beta_e = 0.2$, environmental contamination rate $\sigma = 0.5$, and initial infected-node proportion $p_0 = 0.1$}
    \end{subfigure}
    \hfill
    \begin{subfigure}[t]{0.45\textwidth}
        \centering
        \includegraphics[width=\linewidth]{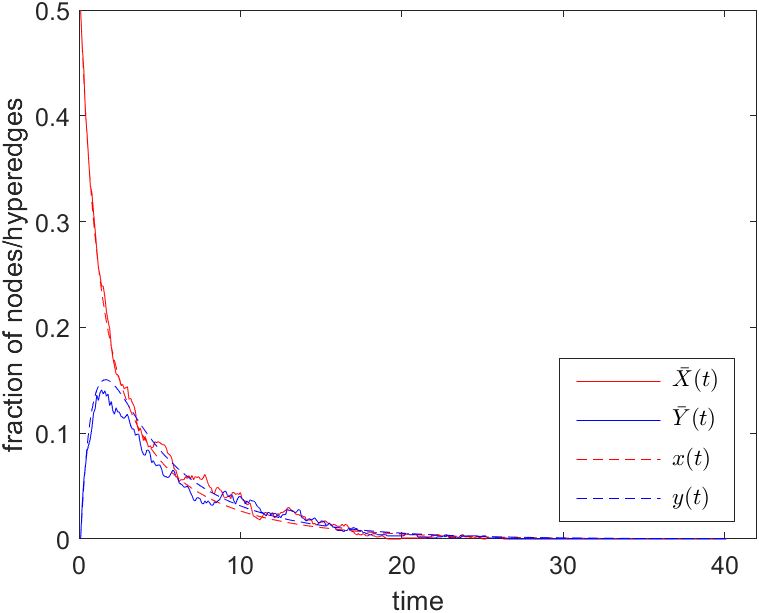} 
        \caption{Disease-free scenario: individual infection rate $\beta_d = 0.03$, environmental infection rate $\beta_e = 0.05$, environmental contamination rate $\sigma = 0.25$, and initial infected-node proportion $p_0 = 0.5$}
    \end{subfigure}
    \caption{Proportions of infected nodes (red) and contaminated hyperedges (blue), averaged over 10 simulations, from the individual-level stochastic model \eqref{eq:node_infected}--\eqref{eq:h_contaminated} (solid curves) and {proportions of infected nodes (red) and contaminated hyperedges (blue) in the} mean-field model \eqref{eqn:meanfield} (dashed curves) on an ER hypergraph. 
    The recovery rates of the nodes and hyperedges are $\gamma = \delta = 1$. 
    In (a), we show results for an endemic scenario. In (b), we show results for a disease-free scenario.
    }
    \label{fig:meanfield_ER}
\end{figure}

\begin{figure}
    \centering
    \begin{subfigure}[t]{0.45\textwidth}
        \centering
        \includegraphics[width=\linewidth]{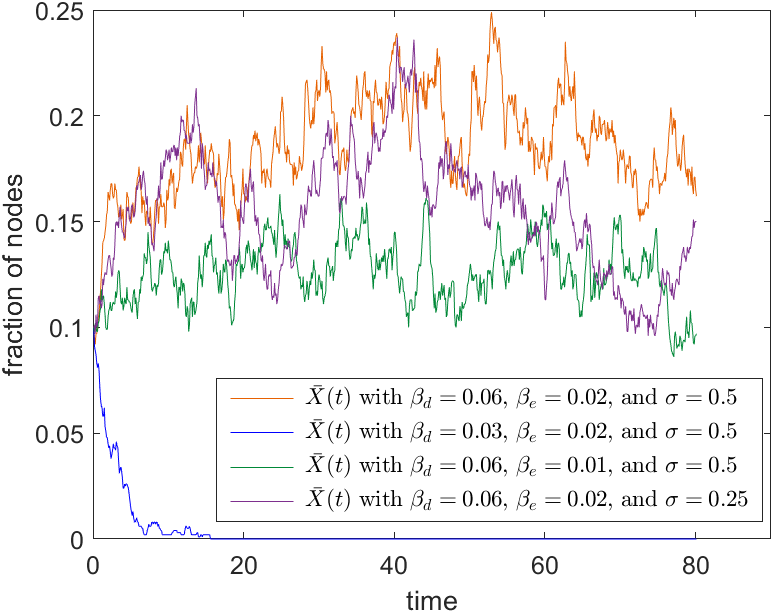} 
        \caption{Dyadic transmission mode is dominant} 
    \end{subfigure}
    \hfill
    \begin{subfigure}[t]{0.45\textwidth}
        \centering
        \includegraphics[width=\linewidth]{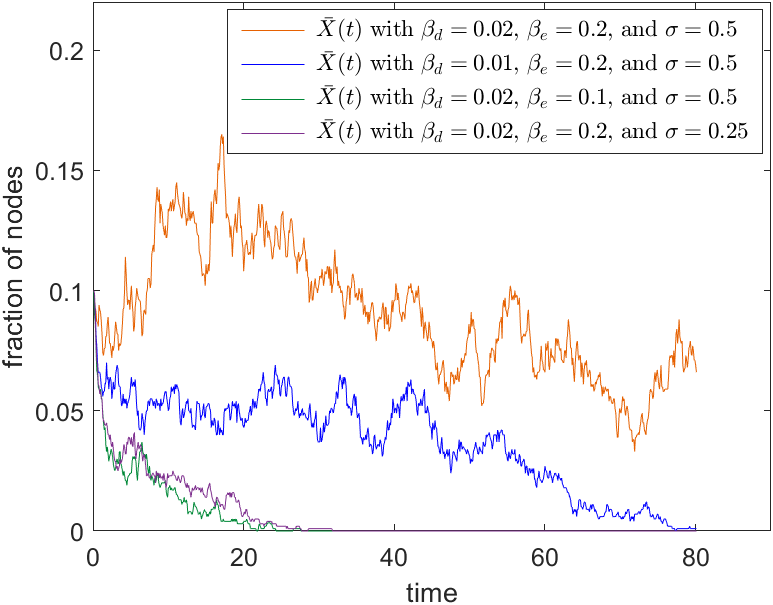} 
        \caption{Polyadic transmission mode is dominant} 
    \end{subfigure}
    \caption{Proportions of infected nodes, averaged over 10 simulations, from the individual-level stochastic model \eqref{eq:node_infected}--\eqref{eq:h_contaminated} on an ER hypergraph
    when (a) the dyadic transmission mode is dominant and (b) the polyadic transmission mode is dominant. The recovery rates of the nodes and hyperedges are $\gamma = \delta = 1$. 
    The initial proportion of infected nodes is $p_0 = 0.1$.
    }\label{fig:R0_ER}
\end{figure}


\subsection{Effect of hypergraph-size distribution on disease dynamics}\label{sec:eff_size}

We now examine the effect of network structure on the disease dynamics. We are motivated by real-life policies to combat and control disease. Many such policies, such as work-from-home orders and restrictions on the sizes of gatherings \cite{gupta2020mandated}, lead to different hyperedge-size distributions.
For instance, restricting the sizes of gatherings to ensure that they are not too large reduces the fraction of large hyperedges and increases the fraction of small hyperedges. Accordingly, we explore
how the 
hyperedge-size distribution affects disease dynamics in both endemic and disease-free scenarios.


\subsubsection{Uniform hyperedge size}

For hypergraphs with a uniform hyperedge size, Remark \ref{remark:size} suggests that hyperedge size has a strong impact on the basic reproduction number $R^{\mathrm{ru}}_0$. For example, in the mean-field model \eqref{eqn:meanfield} on regular $(2,s)$-uniform hypergraphs, doubling the hyperedge size necessitates quadrupling
a recovery rate $\gamma$ or $\delta$ or decreasing an infection rate $\beta_e$ or $\sigma$ by a factor of four to attain the same $R^{\mathrm{ru}}_0$. 

We now perform simulations to examine this impact
for hypergraphs that have uniform hyperedge sizes but are not regular.
We generate one ER hypergraph of each of two types: hypergraph S4 has 100 nodes, 1000 edges, and 100 hyperedges of size 4; and hypergraph S8 has 100 nodes, 1000 edges, and 100 hyperedges of size 8. In Figure \ref{fig:x2size}, we show the proportion of infected nodes in simulations of the stochastic model \eqref{eq:node_infected}--\eqref{eq:h_contaminated} on these two hypergraphs.
We choose the same parameter values for hypergraphs S4 and S8, except for varying one parameter for hypergraph S8.
We observe that decreasing the environmental infection rate $\beta_e$ by a factor of four or quadrupling the enviromental recovery rate $\delta$
leads to hypergraph S8 yielding roughly similar values of the proportion of infected nodes as
in hypergraph S4. 
This observation is consistent with our observation in Remark \ref{remark:size}.

\begin{figure}
    \centering
    \begin{subfigure}[t]{0.45\textwidth}
        \centering
        \includegraphics[width=\linewidth]{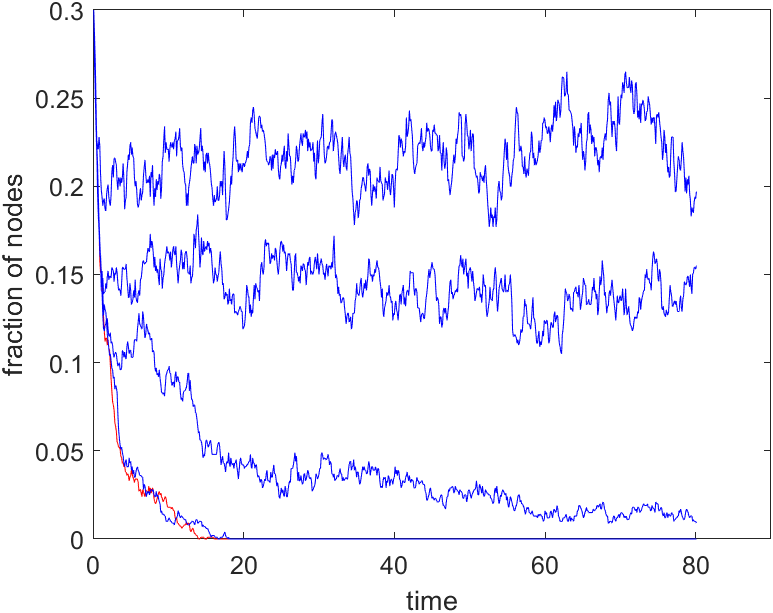} 
        \caption{Varying the environmental infection rate $\beta_e$ in simulations on the hypergraph S8} 
    \end{subfigure}
    \hfill
    \begin{subfigure}[t]{0.45\textwidth}
        \centering
        \includegraphics[width=\linewidth]{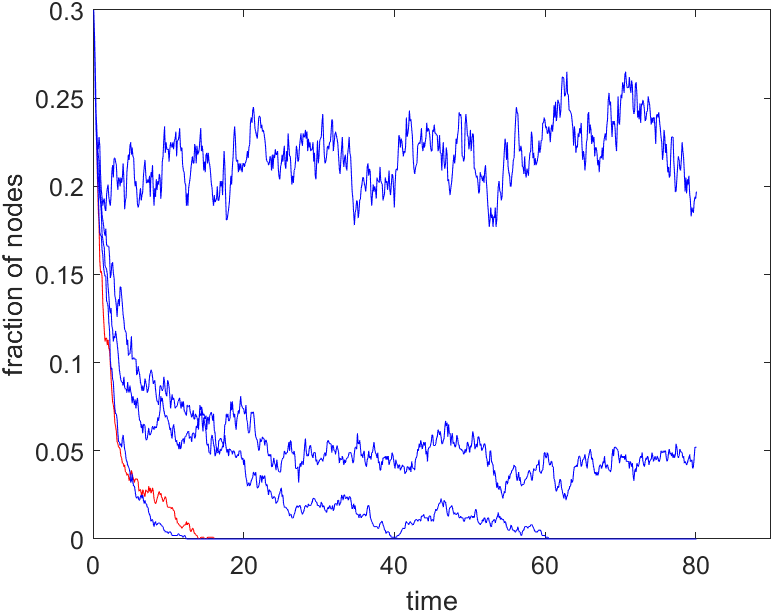} 
        \caption{Varying the environmental recovery rate $\delta$ in simulations on the hypergraph S8}
    \end{subfigure}
    \caption{Proportions of infected nodes, averaged over 10 simulations, in the individual-level stochastic model \eqref{eq:node_infected}--\eqref{eq:h_contaminated} on the hypergraphs S4 (red) and S8 (blue). For our simulations on hypergraph S4, we use the parameter values
    $\beta_d = 0.01$, $\beta_e = 0.1$, $\sigma = 0.5$, $\gamma = \delta = 1$, and $p_0 = 0.3$. (a) For our simulations on the hypergraph S8, we use the same parameter values as for hypergraph S4 except for the environmental infection rate $\beta_e$. From top to bottom, the values of $\beta_e$ are $0.1$, $0.075$, $0.05$, and $0.025$. (b) For our similations on the hypergraph S8, we use the same parameter values as for hypergraph S4 for all parameters except for the 
    environmental recovery rate $\delta$. From top to bottem, the values of $\delta$ are $1$, $2$, $3$, and $4$.}\label{fig:x2size}
\end{figure}


\subsubsection{Nonuniform hyperedge sizes}

We now consider hypergraphs with nonuniform hyperedge sizes and examine how the hyperedge-size distribution affects the dynamics of the stochastic model \eqref{eq:node_infected}--\eqref{eq:h_contaminated}. To understand the impact of the hyperedge-size distribution on the disease dynamics, we generate one realization for each of the three types of ER hypergraphs with different hyperedge-size sequences and simulate \eqref{eq:node_infected}--\eqref{eq:h_contaminated} on them.

Each of the three ER hypergraphs that we generate has $N = 100$ nodes and $M_d = 1000$ edges. Hypergraph H1 has 100 hyperedges of size 8, hypergraph H2 has 50 hyperedges of size 4 and 50 hyperedges of size 12, and hypergraph H3 has 80 hyperedges of size 3 and 20 hyperedges of size 28. These values ensure that all three hypergraphs have a mean hyperedge size of 8. 

Let $\bar{X}_1(t)$, $\bar{X}_2(t)$, and $\bar{X}_3(t)$ denote the proportions of infected nodes in the stochastic model \eqref{eq:node_infected}--\eqref{eq:h_contaminated} on H1, H2, and H3, respectively. Additionally, let $\bar{Y}_1(t)$, 
$\bar{Y}_2(t)$, and $\bar{Y}_3(t)$ denote the proportions of contaminated hyperedges in the stochastic model \eqref{eq:node_infected}--\eqref{eq:h_contaminated} on H1, H2, and H3, respectively.

In Figure \ref{fig:non-unif-size_EN}, we show the proportions of infected nodes and contaminated hyperedges in an endemic scenario (i.e., when we choose the parameters so that $R^{\mathrm{ru}}_0$, which has the value 2.96, for the mean-field model \eqref{eqn:meanfield} is sufficiently above 1) for the three hypergraphs H1, H2, and H3. 
Although the proportion of contaminated hyperedges for the hypergraph H3 seems to be slightly smaller than the proportions for hypergraphs H1 and H2, we do not observe any noticeable differences in the endemic levels for the three hypergraphs.

In Figure \ref{fig:non-unif-size_DF}, we show the proportions of infected nodes and contaminated hyperedges 
when we choose the parameters so that $R^{\mathrm{ru}}_0 = 1.04$ for the mean-field model \eqref{eqn:meanfield} 
for the three hypergraphs H1, H2, and H3. We observe that 
the disease becomes extinct much more slowly on the hypergraph H3 than on hypergraphs H1 and H2.

In Figure \ref{fig:non-unif-size_DF*}, we examine how the disease dynamics changes when we increase
the environmental infection rate $\beta_e$ from $0.02$ (which is the value in Figure \ref{fig:non-unif-size_DF}) to $0.03$. 
The disease still becomes extinct on hypergraphs H1 and H2, but now it
seems to persist on the hypergraph H3. We suspect that the existence of a few very large hyperedges, which correspond to large gatherings of individuals in one location, in hypergraph H3 may be the reason for the disease's persistence.
It is worth pursuing this observation in more detail in future work, including through the systematic consideration of many more hypergraphs with nonuniform hyperedge sizes.

\begin{figure}
    \centering
    \begin{subfigure}[t]{0.45\textwidth}
        \centering
        \includegraphics[width=\linewidth]{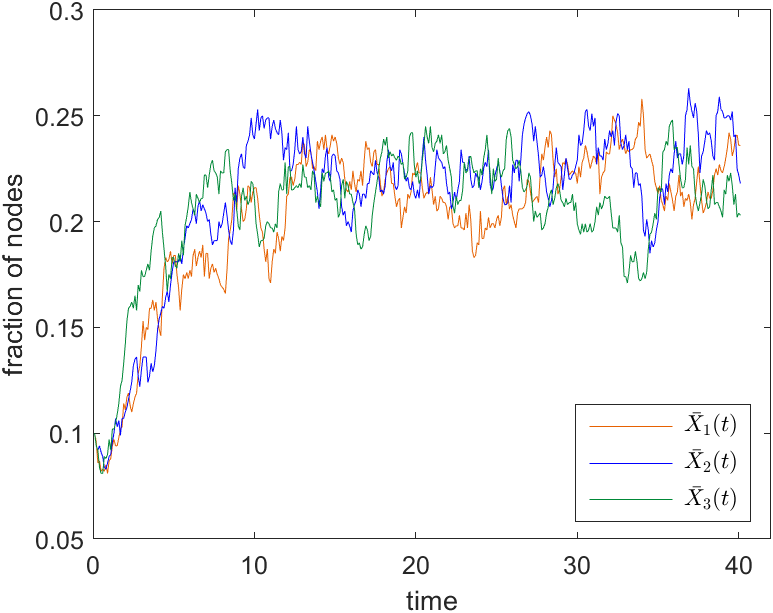} 
        \caption{Proportions of infected nodes} 
    \end{subfigure}
    \hfill
    \begin{subfigure}[t]{0.45\textwidth}
        \centering
        \includegraphics[width=\linewidth]{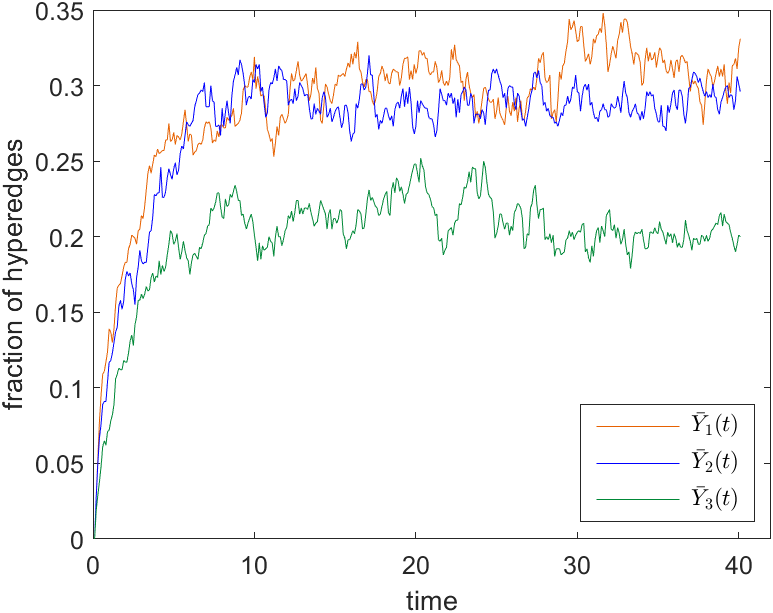} 
        \caption{Proportions of contaminated hyperedges} 
    \end{subfigure}
    \caption{Proportions of (a) infected nodes and (b) contaminated hyperedges, averaged over 10 simulations, in the individual-level stochastic model \eqref{eq:node_infected}--\eqref{eq:h_contaminated} on hypergraphs H1 (blue), H2 (red), and H3 (yellow). The infection rates are $\beta_d = 0.02$, $\beta_e = 0.08$, and $\sigma = 0.5$; and the recovery rates are $\gamma = \delta = 1$. The initial proportion of infected nodes is $p_0 = 0.1$.
    }\label{fig:non-unif-size_EN}
\end{figure}

\begin{figure}
    \centering
    \begin{subfigure}[t]{0.45\textwidth}
        \centering
        \includegraphics[width=\linewidth]{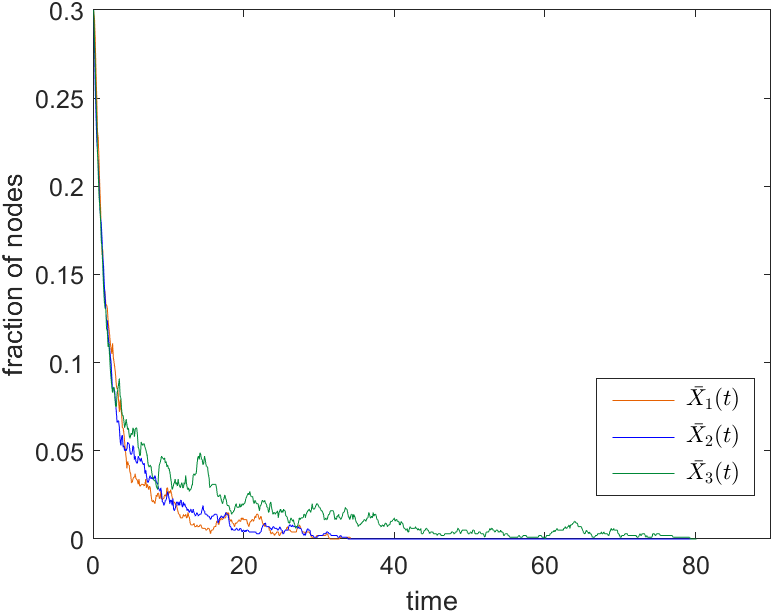} 
        \caption{Proportions of infected nodes} 
    \end{subfigure}
    \hfill
    \begin{subfigure}[t]{0.45\textwidth}
        \centering
        \includegraphics[width=\linewidth]{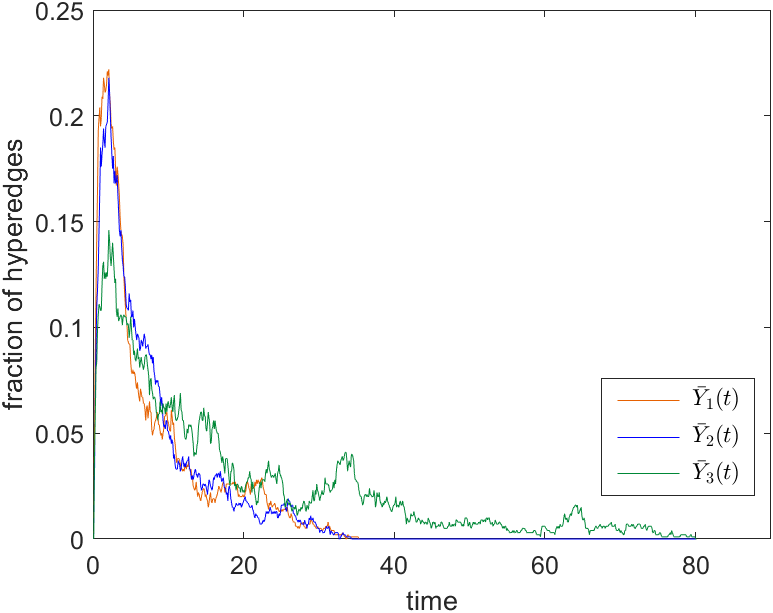} 
        \caption{Proportions of contaminated hyperedges} 
    \end{subfigure}
    \caption{Proportions of (a) infected nodes and (b) contaminated hyperedges, averaged over 10 simulations, in the individual-level stochastic model \eqref{eq:node_infected}--\eqref{eq:h_contaminated} on hypergraphs H1 (blue), H2 (red), and H3 (yellow). The infection rates are 
    $\beta_d = 0.02$, $\beta_e = 0.02$, and $\sigma = 0.5$; and the recovery rates are $\gamma = \delta = 1$. The initial proportion of infected nodes is $p_0 = 0.3$.
    }\label{fig:non-unif-size_DF}
\end{figure}

\begin{figure}
    \centering
    \begin{subfigure}[t]{0.45\textwidth}
        \centering
        \includegraphics[width=\linewidth]{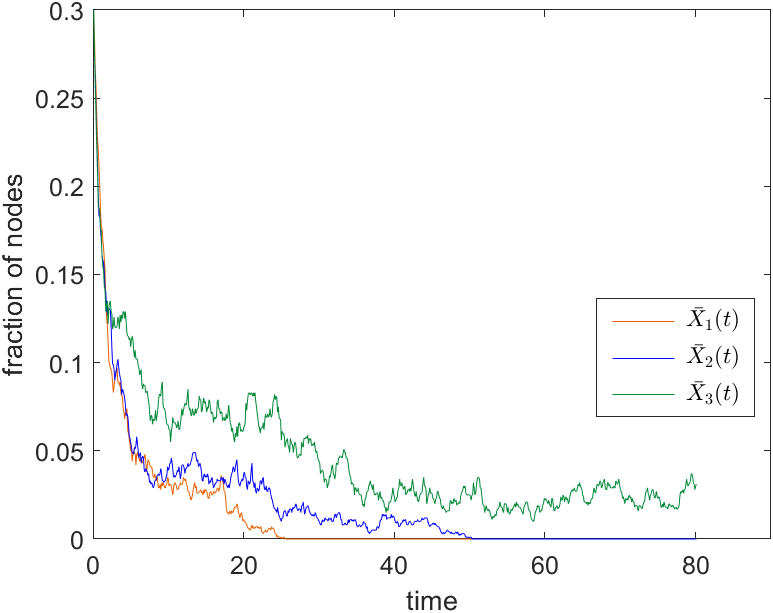} 
        \caption{Proportions of infected nodes} 
    \end{subfigure}
    \hfill
    \begin{subfigure}[t]{0.45\textwidth}
        \centering
        \includegraphics[width=\linewidth]{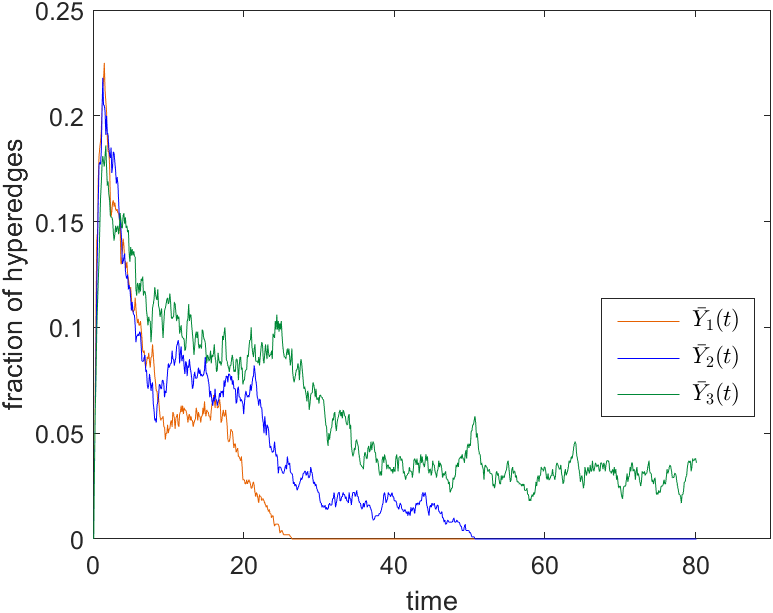} 
        \caption{Proportions of contaminated hyperedges} 
    \end{subfigure}
    \caption{Proportions of (a) infected nodes and (b) contaminated hyperedges, averaged over 10 simulations, in the individual-level stochastic model \eqref{eq:node_infected}--\eqref{eq:h_contaminated} on hypergraphs H1 (blue), H2 (red), and H3 (yellow). The infection rates are 
    $\beta_d = 0.02$, $\beta_e = 0.03$, and $\sigma = 0.5$; and the recovery rates are $\gamma = \delta = 1$. The initial proportion of infected nodes is $p_0 = 0.3$.
    }\label{fig:non-unif-size_DF*}
\end{figure}


\subsection{Effect of nonuniform environmental recovery rates on disease dynamics}\label{sec:eff_delta}

For regular $(2,s)$-uniform hypergraphs, we know from the expression \eqref{eqn:R0} for $R^{\mathrm{ru}}_0$ that if the contribution of the dyadic transmission mode
is $\frac{\beta_d k_d}{\gamma} < 1$, then increasing the environmental recovery rate $\delta$ can reduce $R^{\mathrm{ru}}_0$ 
below $1$ and thereby drive a disease to extinction in the mean-field model \eqref{eqn:meanfield}. Furthermore, if we fix all parameters except for $\delta$ and the hyperedge size $s$, then the extinction threshold value for $\delta$ is
\begin{equation}\label{eqn:delta_threshold}
	\delta_{\text{threshold}} = \frac{\beta_e\sigma g'(0)k_e}{\gamma(1 - \frac{\beta_d k_d}{\gamma})} s \,,
\end{equation}
which depends linearly on $s$.

We now numerically simulate the stochastic model \eqref{eq:node_infected}--\eqref{eq:h_contaminated} to examine the effect of $\delta$ for hypergraphs with nonuniform hyperedge sizes. Motivated by flexible ventilation designs (which one 
adjusts according to the number of individuals in a location) that was proposed in \cite{morawska2021paradigm}, we explore the impact of  nonuniform environmental recovery rates on disease dynamics. In this scenario, a contaminated hyperedge $\ell\in E_e$ becomes uncontaminated with environmental recovery rate $\delta_\ell$.

We generate an ER hypergraph 
with $N = 100$ nodes, $M_d = 1000$ edges, 60 size-4 hyperedges, 30 size-12 hyperedges, and 10 size-20 hyperedges. 
We simulate the stochastic model \eqref{eq:node_infected}--\eqref{eq:h_contaminated} on this hypergraph for three choices of the environmental recovery rates: (1) $\underline{\bm \delta}$, with $\delta_\ell = \delta_{\text{min}} = 1$ for all hyperedges; (2) $\overline{\bm\delta}$, with $\delta_\ell = \delta_{\text{max}} = 2$ for all hyperedges; and (3) $\bm\delta^*$, where we select each $\delta_\ell$ in a way that depends linearly on the size $s_\ell$ of hyperedge $\ell$. In particular, in the third scenario, $\delta_{\ell} = \delta_{\text{min}} + (\delta_{\text{max}}-\delta_{\text{min}})\frac{s_{\ell} - \min_\ell\{s_\ell\}}{\max_\ell\{s_\ell\} - \min_{\ell}\{s_\ell\}}$. From this linear relationship, 
hyperedges of size $4$, $12$, and $20$ have environmental recovery rates of $1$, $1.5$, and $2$, respectively.
In Figure \ref{fig:5}, we show the proportions of infected nodes and contaminated hyperedges for the three choices of enviromental recovery rates. For $\underline{\bm\delta}$, the disease seems to persist, whereas it becomes extinct for both $\overline{\bm \delta}$ and $\bm\delta^*$.

\begin{figure}
    \centering
    \begin{subfigure}[t]{0.45\textwidth}
        \centering
        \includegraphics[width=\linewidth]{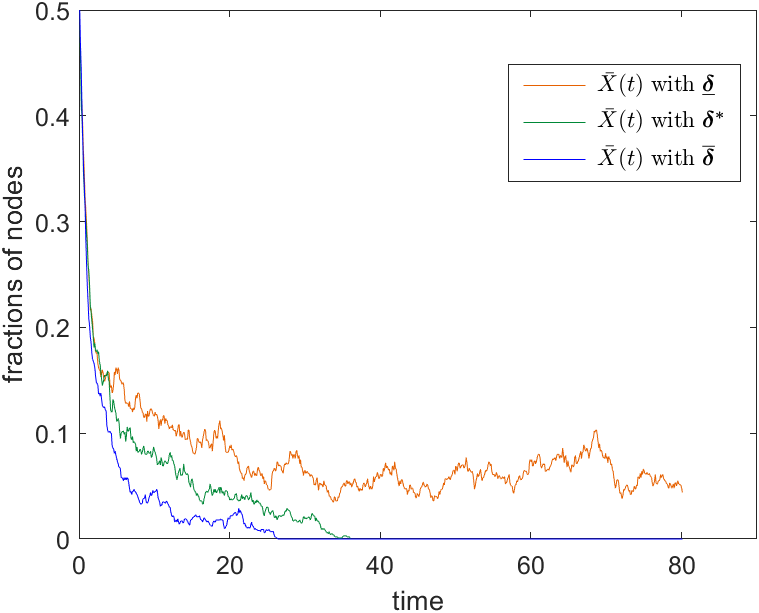} 
        \caption{Proportions of infected nodes}
    \end{subfigure}
    \hfill
    \begin{subfigure}[t]{0.45\textwidth}
        \centering
        \includegraphics[width=\linewidth]{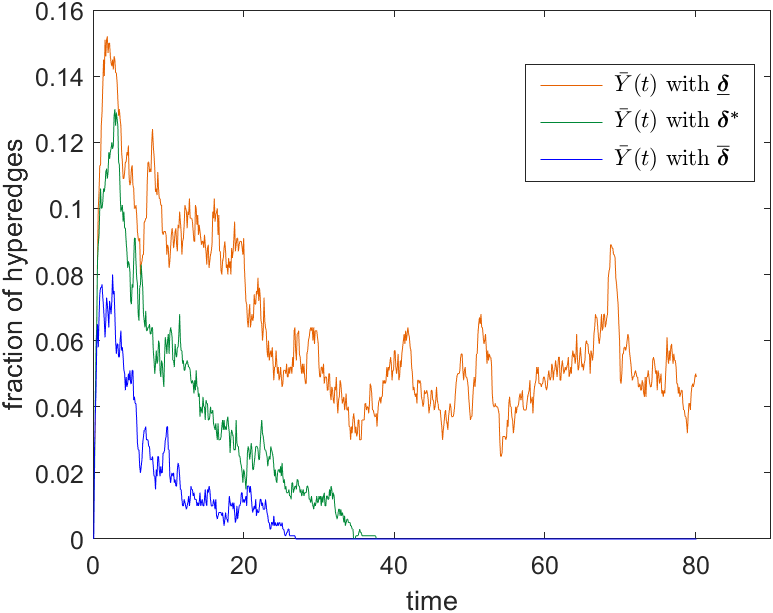} 
        \caption{Proportions of contaminated hyperedges} 
    \end{subfigure}
    \caption{Proportions of (a) infected nodes and (b) contaminated hyperedges, averaged over 10 simulations, in the individual-level stochastic model \eqref{eq:node_infected}--\eqref{eq:h_contaminated} on an ER hypergraph 
    with different distributions of environmental recovery rates. The infection rates are $\beta_d = 0.02$, $\beta_e = 0.08$, and $\sigma = 0.2$; and the individual recovery rate is $\gamma = 1$. The initial proportion of infected nodes is $p_0 = 0.5$.
    }\label{fig:5}
\end{figure}


\section{Conclusions and discussion}\label{sec:discussion}

We studied a susceptible--infected--susceptible (SIS) model on hypergraphs to model the spread of a disease on networks with both dyadic and polyadic interactions. A key novelty of our model is that we distinguish
between the two different modes of transmission. We incorporated this idea into our disease-spread model by assigning state variables
to both nodes and hyperedges. In our model, infected individuals not only spread a disease to other individuals through their social contacts but also ``contaminate" environments by releasing infected aerosols. In turn, contaminated environments can infect the individuals in them. After formulating our model, we derived two mean-field approximations of it and obtained approximate expressions for the basic reproduction number  for complete $(2,s)$-uniform hypergraphs and regular 
 $(2,s)$-uniform hypergraphs. We then showed using numerical simulations that our mean-field description provides
 good approximations of our original individual-level stochastic model
 on regular $(2,s)$-uniform hypergraphs and Erd\H{o}s--R\'enyi (ER) hypergraphs with uniform hyperedge-size distributions. We also performed a variety of other simulations to test the effects of heterogeneous hyperedge-size distributions and hyperedge-recovery rates on the disease dynamics.

Our results have a variety of useful
implications. First, for regular $(2,s)$-uniform hypergraphs, we saw that the basic reproduction number $R^{\mathrm{ru}}_0$ is the sum of contributions from the two distinct transmission modes.
Increasing the recovery rate of individuals can help reduce both terms in $R^{\mathrm{ru}}_0$. Additionally, if the contribution to infection from direct social contacts
is sufficiently small,
then reducing the hyperedge sizes (i.e., the sizes of gatherings) is an effective
way to reduce $R^{\mathrm{ru}}_0$. Second, for hypergraphs with nonuniform hyperedge-size distributions, our simulations of our stochastic model
suggest that having a few very large hyperedges can prolong the time until disease extinction or even allow a disease to persist that otherwise would die out.
Finally, our analysis and numerical simulations both point towards increasing environmental recovery rates (for example, by increasing ventilation and air filtration) as an important
way to drive a disease to extinction. Based on our simulations of our stochastic model,
we hypothesize that an effective way to do this is to increase the hyperedge-recovery rates (i.e., environmental recovery rates)
based on their size,
instead of increasing recovery rates uniformly across all hyperedges.

There are a variety of 
ways to generalize our work. First, one can explore different choices in the model formulation, such as by considering different ways (which we encoded in the sigmoid function $g$) that infected individuals contribute to the disease-transmission risk of an environment and by considering a variety of heterogeneous distributions of infection rates and recovery rates.
It is also relevant to consider different types of hypergraphs, such as generalizations of configuration models and hypergraphs that one constructs from empirical data. Furthermore, we assigned discrete states to hyperedges (which we treated as uncontaminated or uncontaminated), but it may be more realistic to
assign continuous states to hyperedges. In particular, there are many indoor transmission models (see, e.g., \cite{noakes2006modelling,riley1978airborne}) that one can incorporate into a disease-spread model like ours.
It is also relevant to derive mean-field approximations of our disease-spread model for hypergraphs with nonuniform hyperedge sizes. With such an approximation, one can in turn derive analytical results that give insight into the effects of hyperedge-size distributions and nonuniform hyperedge-recovery rates. Such analysis can support our numerical observations in Sections \ref{sec:eff_size} and \ref{sec:eff_delta}.
It is also worthwhile to study the coexistence of competing diseases with our framework. Specifically, it seems interesting to examine when the differences in individual infection rates, environmental infection rates, and environmental contamination rates promote or hinder the coexistence of different diseases.


\section*{Acknowledgements}

We thank Christine Craib
for helpful discussions.


\bibliographystyle{abbrv}
\bibliography{ref-v5}


\end{document}